%% file: main.tex
\documentclass[11pt,letterpaper]{article}
\usepackage{fullpage}
\usepackage[round]{natbib}
\usepackage{amsmath,amsfonts,amssymb,amsthm,  mathrsfs,mathtools}
\usepackage{bbm}
\usepackage{nicefrac}
\usepackage{xcolor}
\usepackage[margin=1in]{geometry}
\usepackage[hyperindex,breaklinks]{hyperref}
\hypersetup{
  colorlinks,
  citecolor=violet,
  linkcolor=blue,
  urlcolor=blue}
\usepackage{dsfont}
\numberwithin{equation}{section}   
\usepackage{booktabs} 
\usepackage[ruled,noend]{algorithm2e} 

\SetAlFnt{\small}
\SetAlCapFnt{\small}
\SetAlCapNameFnt{\small}
\SetAlCapHSkip{0pt}
\IncMargin{-\parindent}

\usepackage[thinc]{esdiff}

\numberwithin{equation}{section}
\usepackage{project-specific-macros}
\usepackage{slivkins-theorems}
\usepackage{suho-macros}

\usepackage{color-edits}
\addauthor{ss}{magenta} 

\title{Delegation with Costly Inspection}

  \newcommand{\country}[1]{#1.}
  \newcommand{\city}[1]{#1}
  \newcommand{\institution}[1]{#1}
  \newcommand{\email}[1]{Email: \texttt{#1}}
  \newcommand{\affiliation}{\thanks}
\author{
 {Mohammad T. Hajiaghayi
 \affiliation{
   \institution{University of Maryland}
   \city{College Park}
   \country{USA}
 \email{hajiagha@umd.edu}
 }}
 \and
 {Piotr Krysta
 \affiliation{
   \institution{Augusta University}
   \city{Augusta}
   \country{USA. Also affiliated with Computer Science Department, University of Liverpool, U.K}
 \email{pkrysta@augusta.edu}
 }}
 \and
 {Mohammad Mahdavi
 \affiliation{
   \institution{University of Maryland}
   \city{College Park}
   \country{USA}
 \email{mahdavi@umd.edu}
 }}
 \and
 \and
 {Suho Shin
 \affiliation{
   \institution{University of Maryland}
   \city{College Park}
   \country{USA}
 \email{suhoshin@umd.edu}
 }}
}

\begin{document}
\date{}
\maketitle
\begin{abstract}
\input{ec2025/abstract}
\end{abstract}


\input{intro}
\input{model}

\input{static-cost}

\section*{Acknowledgements}
    This work is partially supported by DARPA QuICC, ONR MURI 2024 award on Algorithms, Learning, and Game Theory, Army-Research Laboratory (ARL) grant W911NF2410052, NSF AF:Small grants 2218678, 2114269, 2347322, and Royal Society grant IES\textbackslash R2\textbackslash 222170.

\bibliographystyle{plainnat}
\bibliography{ref}

\newpage
\appendix
\input{ec2025/appendix}
\end{document}

%% file: ec2025/abstract.tex
We study the problem of delegated choice with inspection cost (DCIC), which is a variant of the delegated choice problem by Kleinberg and Kleinberg (EC'18) as well as an extension of the Pandora's box problem with nonobligatory inspection (PNOI) by Doval (JET'18).
In our model, an agent may strategically misreport the proposed element's utility, unlike the standard delegated choice problem which assumes that the agent truthfully reports the utility for the proposed alternative.
Thus, the principal needs to inspect the proposed element possibly along with other alternatives to maximize its own utility, given an exogenous cost of inspecting each element.
Further, the delegation itself incurs a fixed cost, thus the principal can decide whether to delegate or not and inspect by herself.

We show that DCIC indeed is a generalization of PNOI where the side information from a strategic agent is available at certain cost, implying its NP-hardness by Fu, Li, and Liu (STOC'23).
In fact, we observe that DCIC becomes significantly more challenging than PNOI and the delegated choice problem in several aspects.
First, neither of (i) running PNOI policy without delegation nor (ii) running a simple delegation mechanism can achieve a constant approximation to the optimal mechanism for DCIC, implying that we need to use both the delegation and inspection for efficient mechanisms.
Furthermore, the standard approaches to PNOI (or Pandora's box problem) for upper bounding the optimal policy in a structured way to obtain algorithms do not easily extend to DCIC. 

Nevertheless, we provide constant approximate mechanisms for DCIC problem.
En route to this result, we first consider a costless delegation setting in which the cost of delegation is free.
We prove that the maximal mechanism over the pure delegation with a single inspection and an PNOI policy without delegation achieves a $3$-approximation for DCIC with costless delegation, which is further proven to be tight.
These results hold even when the cost comes from an arbitrary monotone set function, and can be improved to a $2$-approximation if the cost of inspection is the same for every element.
We extend these techniques by presenting a constant factor approximate mechanism for the general setting for rich class of instances.
    


    

%% file: intro.tex
\section{Introduction}

Delegation is one of the most fundamental economic applications and frameworks, describing the process of delegating a decision making, search, or optimization task to an expert with information benefit. 
Indeed, a significant number of real-world scenarios can be seen as acts of delegation.
For instance, the government delegates the role of producing an item and trade to revenue maximizing intermediary in revenue maximization for auctions~\citep{manelli2007multidimensional,li2013revenue,daskalakis2013mechanism} and mediated bilateral trade~\citep{vcopivc2008robust,eilat2021bilateral,kuang2023profit,hajiaghayi2024gains}. Online content platforms delegate the role of contents generation to the creators~\citep{yao2024rethinking,shin2022multi,dai2024can,jagadeesan2024supply}, and even users delegate information seeking procedure to generative AI~\citep{immorlica2024generative}.

The delegated choice problem, introduced by~\cite{holmstrom1980theory,armstrong2010model}, is one of the most fundamental settings in delegation literature to study the loss of moral hazard arising from information asymmetry in a principal-agent setting.
Formally, there exists a set of $n$ alternatives, and the decision maker (principal) wants to commit to an alternative, but has no information on the ex-post quality of them.
Instead, the principal delegates the decision making process to an expert (agent) who can observe the values of the alternatives.
The agent, however, realizes a misaligned utility for the decision making process, having an incentive to mislead the outcome to maximize his own utility, rather than the principal's.
To prevent this, the principal commits to a screening mechanism that screens out an alternative once the utility of the alternative that the agent has submitted is smaller than a predetermined threshold, using the distributional information on each alternative's utility.
\cite{kleinberg2018delegated} reveal an interesting connection between threshold-based screening mechanisms and non-adaptive algorithms for the prophet inequality~\citep{krengel1987prophet,samuel1984comparison}.

The fundamental assumption in the standard delegated choice problem is that the agent truthfully reports the corresponding utility of the proposed alternative.
For instance, this could be an easily verifiable lie once the decision making process is over and the alternative is implemented, or it may backfire on the agent's reputation status.
However, in many real-world applications, the cost of verifying the utility of the outcome might be expensive, \eg delegated decision of a governmental policy, or if the delegation happens only as a one-off interaction, so that the agent might have incentive to misreport the utility of the proposed alternative.
In this context, we consider a variant of the delegated choice problem, termed \emph{delegated choice with inspection cost} (DCIC), where the principal can further inspect the alternative at a certain cost upon the agent's signal and the agent might not truthfully report the utility of the proposed alternative.

We first observe that DCIC is a strict generalization of the recently introduced variant of the Weitzman's box problem by~\cite{weitzman1978optimal}, called Pandora's box with nonobligatory inspection (PNOI)~\citep{doval2018whether,beyhaghi2019pandora}.
In the PNOI problem, there exists a set of boxes (alternatives) with a private reward for each box, the distribution of which is publicly available to the searcher.
The searcher wants to commit to one of the boxes to maximize the reward.
To this end, the searcher sequentially either (i) opens a box $i$ at fixed cost $c_i$ to observe the reward or (ii) selects a box that might be opened or closed, and finishes the searching process.
Different from the standard Pandora's box problem with obligatory inspection which has a simple optimal index-based policy, PNOI problem is proven to be NP-hard~\cite{fu2023pandora} even if the support of each of the reward distributions is of cardinality three. 
Indeed, we observe that the PNOI problem reduces to an instance of DCIC problem where the cost of delegation is significantly larger than values obtainable from any alternative, which implies NP-hardness of the DCIC problem.

In this paper, we formally introduce and study the DCIC problem as extensions of both the delegated choice problem and PNOI problem, and provide the first constant approximate mechanisms for several versions.
We first show that neither pure delegation policy nor pure PNOI policy can achieve a constant approximation to the optimal policy, implying that a combination of both the delegation and inspection is necessary to achieve constant approximation.
Correspondingly, we consider a setting with costless delegation where the cost of the delegation is zero, and obtain a $3$-approximate mechanism, and further prove that this is essentially tight.
We show that this result extends to a combinatorial cost setting~\cite{berger2023pandora}, where the cost function is an arbitrary monotone set function.
We finally obtain a constant factor approximation mechanism for the general setting of DCIC for rich class of instances.

As a conceptual contribution, our work connects the Pandora's box problem and the delegated choice problem, complementing the intimate connection between the delegated choice problem and prophet inequality by~\cite{kleinberg2018delegated}, one of the most central problems in online algorithms.
Furthermore, to the best of our knowledge, our work initiates the first study of Pandora's box with side information given by the strategic agent's signal, which might be of independent interest.

We first briefly summarize our technical results in Section~\ref{sec:sum-results} along with a concise version of problem setup.
Further related works can be found in Section~\ref{sec:related}, and formal problem setup is detailed in Section~\ref{sec:model}.
Our main results are presented in Section~\ref{sec:results}, where we defer the proofs of the propositions to the appendix due to page limit.

\section{Our results and techniques}\label{sec:sum-results}
We introduce here our main results and ideas for their proofs.
To this end, we briefly introduce our problem setup. Formal model description can be found in Section~\ref{sec:model}.
In the DCIC problem, we have a set of alternatives\footnote{We often refer to an alternative as a solution, an element, or a box (as in the Pandora's box problem).} $[n] = \{1,2,\ldots, n\}$,  and the principal needs to select one of them.
Each alternative $i$ is equipped with independent nonnegative random variables $X_i$ and $Y_i$, denoting the principal's and the agent's utilities, respectively.
Let $D^X_i$ and $D^Y_i$ denote the distribution associated with $X_i$ and $Y_i$, respectively, and $D^X$ be the product distribution of $(X_i)_{i \in [n]}$, and we similarly define $D^Y$.
Let $\cD^X$ and $\cD^Y$ be the respective families of all possible such distributions over $\R_{\ge 0}^n$.
Each alternative $i$ is also equipped with a deterministic inspection cost $c_i$.
The principal is aware of the distribution $D^X_{i}$ and the cost $c_i$ for $i \in [n]$.
At each time, the principal can decide either to inspect an uninspected alternative, select an alternative and realize its utility, or stop the search.
If the principal inspects a set of alternatives $I$ and commits to alternative $k \in [n]$, the principal realizes the utility of $X_k - \sum_{i \in I}  c_i$.
Let $\cP$ be a set of all possible search (inspection) policies.

Without the presence of the agent, the principal's search procedure to maximize the utility can be exactly framed as the Pandora's box problem with nonobligatory inspection (PNOI)~\cite{doval2018whether}.
On the other hand, the principal can consult an agent who can observe the realization of every alternative.
Importantly, we assume that the principal needs to pay the exogenous cost of $c_{\Del}$ to delegate to the agent and receive the feedback.
The agent, however, realizes the utility of $Y_k$ if the principal decides to select alternative $k$, \ie having a misaligned utility.
Let $\Sigma$ be the set of all possible signals by the agent; \eg in~\cite{kleinberg2018delegated}, this might be a pair of a proposed alternative and a corresponding utility.
Then, given the principal's committed mechanism, the agent best-responds by sending a signal $\sigma \in \Sigma$ that maximizes the agent's utility.
Overall, the game proceeds as follows:
\begin{enumerate}
    \item The principal commits to a mechanism $\Mec: \Sigma \to \cP$ that maps the agent's signal to an inspection policy.
    \item The agent observes the realization of $(X_i,Y_i)_{i \in [n]}$ and best-responds with signal $\sigma$ given $\Mec$.
    \item The alternative (including null) is determined by the outcome of the search policy $\Mec(\sigma)$.
\end{enumerate}

We are interested in designing an \emph{agent-oblivious} mechanism that has no information about the agent's utility distributions $D^Y$.
Formally, the principal's expected utility for mechanism $\Mec$ can be written as
\begin{align*}
    \Ex{\Mec} = \Ex{\parans{\sum_{i \in [n]} \bA_i X_i - \bI_i c_i} - \bI_{\Del} c_{\Del}},
\end{align*}
where $\bA_i$ is the indicator random variable denoting whether $i$ is selected, $\bI_i$ denotes an indicator variable of whether $i$ is inspected, and $\bI_{\Del}$ indicates whether the principal decides to delegate or not.
Furthermore, the worst-case utility of the principal for an oblivious mechanism $\Mec$ is defined over the worst-case possible agent's distribution, written as
\begin{align*}
    \Exu{D_X}{\Mec} = \min_{D_Y \in \cD_Y}\Exu{D_X,D_Y}{\parans{\sum_{i \in [n]} \bA_i X_i - \bI_i c_i} - \bI_{\Del} c_{\Del}}.
\end{align*}
Given an optimal mechanism $\Opt$ that maximizes the above quantity over the space of feasible mechanisms, $\Mec$ is an $\alpha$-approximation if
\begin{align*}
    \alpha \cdot \min_{D_X \in \cD_X}\frac{\Exu{D_X}{\Mec}}{\Exu{D_X}{\Opt}} \ge 1,
\end{align*}
\ie the worst-case ratio between the optimal oblivious mechanism and $\Mec$. 

As discussed, DCIC can be viewed as a version of PNOI problem where the side information from a strategic agent is possible at certain costs.
On one hand, one can regard such information from the agent as another box that changes the distribution of the other boxes.
On the other hand, the correlation structure on how the agent's signal changes the other boxes' distributions is very unclear, which makes our problem significantly harder and fundamentally different from the Pandora's box problem with correlated rewards.

Before presenting our main results, it is straightforward to see that DCIC problem is NP-hard, 
since the PNOI problem reduces to DCIC by setting $c_{\Del}$ sufficiently large.
\begin{proposition}
    DCIC problem is NP-hard.
\end{proposition}

\subsection{First attempt: PNOI policy or SPMI}
A natural attempt that one might try first might be to see if either of the delegation-only (without further inspection on elements not proposed by the agent) mechanism or inspection-only (without delegation), \ie a PNOI policy, mechanism is sufficient to achieve any constant approximation.

We observe that the inspection-only mechanism cannot achieve constant approximation as follows, which necessitates the a smarter mechanism beyond the delegation.

\begin{customprop}{\ref{prop:no-del-fails}}
    There exists a problem instance in which the optimal direct inspection policy cannot achieve a constant approximation ratio.
\end{customprop}
Intuitively, one can consider $n$ identical boxes with the same distribution supported on two values $v > 0$ and $0$.
For the inspection-only policy without delegation, the decision maker only needs to determine how many alternatives to inspect before it observes $v$ due to the symmetry of the alternatives.
On the other hand, one may set a simple threshold-based delegation mechanism for which the agent is required to submit a single alternative and it is accepted if and only if it induces the principal's utility at least the specified threshold.
If we appropriately set the probability for $v$, there exists only one alternative realizing $v$ with high probability.
In this case, the inspection-only policy necessarily suffers significant amount of inspection costs, whereas the delegation mechanism, with any nonzero threshold would enjoy the utility of $v$ with just a single inspection.

For the delegation-only mechanism, we consider the single-proposal mechanism (SPM) proposed by~\cite{kleinberg2018delegated}, which achieves $2$-approximation to the first best outcome $\Ex{\max_{i \in [n]} X_i}$ under the assumption with truthful report.
Since the original SPM does not inspect the proposed solution, which incentivizes the misreport of the utility of the proposed alternative by the agent, we consider a natural variant of SPM, dubbed \emph{single-proposal mechanism with inspection} (SPMI), that further inspects the proposed alternative.
Further, the following result shows that the single-proposal mechanism fails to achieve any constant approximation ratio as well.
\begin{customprop}{\ref{prop:spm-fails}}
    There exists a problem instance in which the SPMI cannot be $\alpha$-approximate for any $\alpha > 0$.
\end{customprop}
The result follows from the intuition that if the inspection cost is sufficiently large, then selecting an alternative without any inspection dominates the SPMI.

\subsection{Next attempt: using ideas of~\cite{kleinberg2016descending,beyhaghi2019pandora}}
On the other hand, one may wonder whether we could smartly upper bound the optimal mechanism in an efficient manner, to build our intuition towards an efficient mechanism.
Indeed, such approaches are quite common in the original Pandora's box problem~\cite{weitzman1978optimal}, and even for its variant~\cite{beyhaghi2019pandora}.
These techniques, however, turn out to not easily extend to DCIC, as will be shown in what follows.

Let us first recap the standard approaches.
In the standard Pandora's box problem, the principal's utility using the optimal policy can be upper bounded by a tuple of random variables $(\kappa_i)_{i \in [n]}$ defined by $\kappa_i = \min(X_i, \sigma_i)$ where $\sigma_i$ satisfies
\begin{align*}
    \Ex{(X_i - \sigma_i)^+} = c_i.
\end{align*}
The quantity $\sigma_i$ is usually called \emph{cap}.
Then, \cite{kleinberg2016descending} proved that 
\begin{align}
    \Ex{\bA_i X_i - \bI_i c_i} \le \Ex{\bA_i \kappa_i}.\label{eq:beyhaghi}
\end{align}
The proof follows from decomposing $X_i$ into $\kappa_i + b_i$ such that
\begin{enumerate}
    \item If $X_i > \sigma_i$, then \emph{capped value} $\kappa_i = \sigma_i$ and \emph{bonus} $b_i = X_i - \sigma_i$.
    \item Otherwise, $\kappa_i = X_i$, \ie capped value is claimed.
\end{enumerate}
Then, one can easily see that the construction of the bonus $b_i$ satisfies $c_i = \Ex{b_i}$, \ie the expected amount of bonus marginalizes the cost of the inspection.
Thus, a net utility from each box can be viewed as exactly the capped value.

Correspondingly, this inspires a very natural greedy-like search policy to match the described upper bound: opens the boxes in a decreasing order of $\sigma_i$, and stops when the largest observed $X_i$ exceeds the highest remaining cap $\sigma_j$ and select such $i$. This policy is shown to be exactly the optimal and achieving the above upper bound.

For the PNOI problem, \cite{beyhaghi2019pandora} analogously follows this intuition to upper bound the optimal mechanism's utility in a structured way.
Indeed, one can modify the random variable $\kappa_i$ to $\tilde{\kappa}_i$ so that $\tilde{\kappa}_i = \kappa_i$ if $i$ is inspected and $\tilde{\kappa}_i = \Exu{v_i \sim D^X_i}{v_i}$ if $i$ is not inspected, \ie
\begin{align*}
    \tilde{\kappa}_i = 
    \begin{cases}
        \kappa_i &\mbox{ if } \bI_i = 1 \\
        \Ex{X_i} &\mbox{ if } \bI_i = 0
    \end{cases}
\end{align*}
and maintain the upper bound by replacing $\kappa_i$ into $\tilde{\kappa}_i$, \ie $\Ex{\bA_i X_i - \bI_i c_i} \le \Ex{\bA_i \tilde{\kappa}_i}$.
Intuitively, such an upper bound easily extends from the standard setting since the reward of an uninspected box can be simply upper bounded by the expected value of its reward. Given this, a simple upper bound for the utility of the optimal PNOI algorithm can be found based on the solution to Pandora's box problem and the fact that the maximum utility gained from choosing uninspected solutions, $E[\sum_iX_i\bA_i\cdot\Ind{\bI_i = 0}]$, is at most $\max_{i \in [n]} \Ex{X_i}$.

On the other hand, in our setting, even the uninspected solutions' utility can be significantly larger than $\Ex{X_i}$ depending on the principal's mechanism and the information gained from the agent's signal, formally stated as follows.
\begin{customprop}{\ref{prop:info}}
    There exists a problem instance where the ratio of utility gained from uninspected solutions for some mechanisms and corresponding agent's best-response signal, \ie $E[\sum_iX_i\bA_i\cdot\Ind{\bI_i = 0}]$ to the maximum expected utility of any solution $\max_{i \in [n]} \Ex{X_i}$ is $\Omega(n)$.
\end{customprop}
Thus, it is highly unclear how to upper bound the optimal mechanism's utility for DCIC unlike the PNOI problem or the standard Pandora's box problem.

\subsection{Costless delegation}
We will now outline our approach that will lead to constant approximate mechanisms for DCIC.
Recall that in previous attempts, we observe that it is crucial to (i) obtain a reasonable upper bound on the optimal mechanism, and (ii) accordingly design an efficient mechanism that smartly exploits both the delegation and further inspection.
To this end, we first consider a simpler case where the cost of delegation is zero, which will serve as a stepping stone in the general setting in Section~\ref{sec:con-costly}.

We first formally provide our main result in the costless delegation setting.
\begin{customthm}{\ref{thm:costless-3pax}}\label{thm:intro-costless}
    There exists a mechanism that achieves an approximation ratio of 3 for delegated choice with inspection costs problem if the delegation cost is zero.
\end{customthm}
Let us briefly summarize our technical steps to prove this theorem. Clearly, the principal is always not worse off by choosing to delegate and receive the agent's signal since she can simply ignore the signal in any case.

Thus, we first aim to establish an upper bound for the utility obtained from an optimal mechanism.
To this end, we introduce a notion of \emph{overinspection}.
Formally, an alternative $i$ is selected without overinspection if $$\bA_i = 1, \quad \bU_i=\bigvee_{j:c_j \geq c_i} \bI_j = 0,$$ \ie any alternative with the same or higher inspection cost than $c_i$ has not been inspected.
Then, we can effectively prove that there exists an agent such that the expected utility from selecting a solution without overinspection is upper bounded by the expected value, \ie
\begin{align}
    \sum_{i \in [n]}\Ex{X_i \bA_i \cdot \Ind{\bU_i = 0}} \le \max_{i \in [n]}\Ex{X_i}.\label{eq:lemma3-4}
\end{align}
At the heart of this proof is to construct a best-response from an agent such that if a solution is chosen without overinspection, it is always the lowest-cost solution that can be selected without overinspection given the principal's mechanism.
Note that~\eqref{eq:lemma3-4} exhibits a similar structure to~\eqref{eq:beyhaghi}, but has subtle differences on both sides which allow us to further upper bound the utility of any policy.

Accordingly, we provide an upper bound on the utility obtained without overinspection from this best-response by considering another response that leads to the selection of a non-overinspected solution whenever possible. Finally, we argue that whether a solution is chosen without overinspection given this response is dependent only on the utility of alternatives with strictly lower inspection costs, and therefore independent of the utility of the selected alternative itself. This independence shows that we essentially gain no information on the utility of the selected alternative and the expected utility gained will match the alternative's ex-ante utility $\Ex{X_i}$.

Combining these ideas, in Lemma~\ref{lm:opt-fullub}, we can further upper bound the utility obtained by any mechanism with the following quantity:
\begin{align}
    \Ex{\Opt} \le \max_{i\in[n]} \Ex{X_i} + \Ex{\max_{i\in[n]} (X_i - c_i)^+},\label{eq:11031710}
\end{align}
where $\Opt$ is an arbitrary (optimal) mechanism's expected utility.

Given the upper bound on any feasible mechanism, it remains to derive a mechanism that is a constant approximate to the quantity in~\eqref{eq:11031710}.
We consider a simple form of \emph{maximal mechanism} that takes the one with the higher expected utility from two mechanisms.
Those two mechanisms are (i) $\Mec_1$: the one that uses the SPMI, and (ii) $\Mec_2$: the one that simply selects the solution with the largest expected utility without any inspection.

We first prove that $\Mec_1$ achieves at least half of the last term in the right hand side of~\eqref{eq:11031710} in Lemma~\ref{lm:delegation}.
Then, we can prove that the maximal mechanism described above achieves the maximum between the following two quantities:
\begin{align*}
    \Ex{\Mec} \ge \max \left( \max_{i \in [n]}\Ex{X_i}, \frac{1}{2} \Ex{\max_{i \in [n]} (X_i - c_i)^+}\right)
\end{align*}
Combining the results, we obtain the $3$-approximate mechanism in Theorem~\ref{thm:intro-costless}.
Further, this bound is shown to be tight for our mechanism, \ie the mechanism above is exactly $3$-approximation, as stated in Theorem~\ref{thm:tightness}.

Interestingly, we prove that our analysis seamlessly carries over to the \emph{combinatorial cost} setting, where there exists a monotone cost (set) function $c: 2^{[n]} \to \R_{\ge 0}$ that maps a set of selected elements to the overall cost that might not be additive over each item~\cite{berger2023pandora}.
This is formally stated in Corollary~\ref{cor:combinatorial}.

Further, if the cost of the delegation is the same, \ie $c_i = c_j$ for $i,j \in [n]$, we improve the approximation ratio to be $2$ in Proposition~\ref{prop:identical}.
This also holds in the combinatorial cost setting, if  $c(\{i\}) = c(\{j\})$ for every $i,j \in [n]$.


\subsection{Costly delegation}\label{sec:con-costly}
In the costly delegation setting, we prove the following theorem:
\begin{customthm}{\ref{thm:main}}
    Let $\Mec_P$ be the optimal mechanism for the PNOI problem. If $c_{\Del} = \alpha \Ex{\max_{i \in[n]} (X_i -c_i)^+}$ for some $\alpha < 1/2$, then there exists a $(3-4\alpha)/(1-2\alpha)$ approximate mechanism.
    If $c_{\Del} \ge \Ex{\max_{i \in [n]} (X_i - c_i)^+} - 2\Ex{\Mec_P}$, then there exists a $2-\eps$ approximate mechanism for any $\eps > 0$.
\end{customthm}
This theorem essentially characterizes the constant factor approximability when the cost of delegation is relatively large or small.
One immediate corollary of this theorem is that if $\Ex{\Mec}  \ge (1+\alpha)/2 \cdot \Ex{\max_{i \in [n]} (X_i -c_i)^+}$ for some constant $\alpha$, then there exists a constant factor approximate mechanism regardless of the delegation cost.

The assumption on the delegation cost is innocuous in a sense that DCIC still remains NP-hard under such assumptions, since the second case of Theorem~\ref{thm:main} allows us to increasing $c_{\Del}$ as desired.
The case where the cost of delegation is moderately large remains a major open problem.

In its proof, we consider a maximal mechanism that either runs (i) PTAS for the PNOI problem by~\cite{beyhaghi2023pandora} and~\cite{fu2023pandora} or (ii) the SPMI that only delegates and inspects the (nonnull) proposed element.
The proof begins with decomposing the optimal mechanism's utility similar to the costless setting, and exploits the $(1-\eps)$ approximability of PTAS and the imposed assumption in a careful manner.
Finally, we provide an improved $2$-approximate mechanism when the cost of inspection $c_i$ is the same for every $i\in [n]$.

\subsection{Related Works}\label{sec:related}

\cite{holmstrom1980theory} initiates a substantial line of research on delegation mechanisms without transfers. They examined a setting where a principal delegates an optimization problem over a compact interval to a single agent and established conditions for the existence of an optimal mechanism for the principal. \cite{alonso2008optimal} expand on these findings by offering a broader characterization of the optimal mechanism.
The discrete choice setting is first explored by \cite{armstrong2010model} in which an agent samples potential solutions from a distribution and optimizes over this set, providing a setup closely aligned with our problem.

Recently, \cite{kleinberg2018delegated} identify its novel connection to the prophet inequality, and analyze several approximation factor to the first-best outcome.
We note here that the second model of~\cite{kleinberg2018delegated}, termed delegated search problem, also reveals its connection to the Pandora's box problem.
However, their model is significantly different from ours since in their model, since the agent suffers the search cost to examine the alternative assuming that the agent truthfully reports the utility afterwards, whereas in our model the principal handles the search cost and the agent also might strategically misreport the utility.
\cite{khodabakhsh2024simple} also tackle the assumption of truthful reports of the proposed option, or alternatively easily verifiable utility of the proposed options, with a stylized model on the principal's and the agent's utility.
However, they consider a slightly different model under which the principal can only rule each option completely in or out based on the distributional knowledge, without relying on additional device such as inspection.
Interestingly, they also show that the approximation to the first-best benchmark is not possible similar to ours, and thus focus on second-best benchmark.
They show that the optimal mechanism is NP-hard to compute, and provide $1/3$-approximate deterministic mechanism.
A number of papers~\citep{bechtel2020delegated,bechtel2022delegated,shin2023delegating,hajiaghayi2024regret} study several other variants of the delegated choice problem from algorithmic perspective, but we will omit the details due to significant difference.

The Pandora's box problem, initiated by~\cite{weitzman1978optimal}, has been also a central problem in theoretical computer science, grounded as one of the most fundamental online algorithm problem.
The optimal policy is surprisingly simple to have an index-based nature.
The most closely related to ours is the variant with nonobligatory inspection introduced by~\cite{doval2018whether},\footnote{In fact, \cite{guha2008information} firstly study the setup of the Pandora's box with nonobligatory inspection (with discrete reward distribution) in the context of multi channel wireless networks, yet \cite{doval2018whether} firstly formalizes the problem as a variant of the Pandora's box problem.} where the searcher can select a box without necessarily opening it.
\cite{beyhaghi2019pandora} propose a class of policy defined as committing policy, and prove that the optimal committing policy achieves $1-1/e$ approximation to the optimal policy.
Simultaneous works by~\cite{beyhaghi2023pandora} and~\cite{fu2023pandora} prove that the nonobliatory inspection problem is NP-complete, and provide a PTAS algorithm based on the framework by~\cite{fu2018ptas}.
Another line of works study several other variants of the standard Pandora's box problem~\citep{singla2018price,bowers2024matching}, however, their setup is largely different from ours.







%% file: model.tex

\section{Problem Setup}\label{sec:model}

\subsection{Delegated Choice}
We first introduce the \emph{delegated choice} problem by \cite{armstrong2010model,kleinberg2018delegated}. 
There is a principal who wants to solve a certain problem, for which there is a set $\om = \{\om_1,\om_2,\ldots, \om_n\} \subseteq \Omega$ of possible pre-specified alternatives (equivalently, solutions, options, or boxes). 
However the principal is not an expert and may want help from an expert agent to choose the right solution.
Without loss of generality, we assume $\Omega$ includes a null element $\perp$ to handle the case when no element is selected in the end, which results in the utility of zero for both the principal and agent.
Each $\om_i$ is independently sampled from a probability distribution $D_{i}$ supported on $\Omega$.
We assume that the set $\om$ of available elements (including the number of solutions) is exogenous to the mechanism designed by the principal, which is standard in literature of the delegated choice problem, \cf \cite{armstrong2010model,kleinberg2018delegated}.
We write $\Omega^*$ to denote the collection of finite sequences over $\Omega$.

The principal and the agent have their own utility functions $x(\cdot),y(\cdot):\Omega \mapsto \R_{\ge 0}$, respectively.
These utility functions might be misaligned, and the agent is interested in maximizing his own payoff, rather than maximizing the principal's. 
There exists an information asymmetry between the principal and the agent such that the principal is not aware of the utilities, but the agent can observe all the utilities $x(\om_i)$ and $y(\om_i)$ for $i \in [n]$.
After the agent observes the elements and utilities therein, he sends a signal, \eg an element, to the principal, and then the principal finally decides whether or not to accept the proposed element.
Without any screening device to restrict the agent's behavior, the agent will obviously send a signal to maximize his own payoff, \ie propose an element that maximizes $y(\om_i)$.
We write $X_i$ and $Y_i$ to denote the random utility $x(\om_i)$ and $y(\om_i)$, respectively, and $D_{X_i}$ and $D_{Y_i}$ to denote the corresponding probability distributions of these utilities.
Let $D_X = \times_{i \in [n]} D_{X_i}$ and $D_Y = \times_{i \in [n]} D_{Y_i}$, and $\cD_X$ be the set of all possible distributions $D_X$ over nonnegative real numbers given $[n]$, and similarly define $\cD_Y$.

To prevent adversarial behavior of the agent, the principal devises a mechanism $\Mec = (\Sigma, g)$ where $\Sigma$ is an arbitrary set of signals and $g : \Sigma \to \Omega$ is an allocation function given the agent's signal. Namely, once the principal commits to a \emph{signaling mechanism} $\Mec$, the agent sends a signal as best-response, and the principal commits to an element based on the sent signal.
We write $X^*$ to denote the set of all possible finite sequences over $X$ for any finite set $X$.
Given a mechanism $\Mec$, the agent's strategy is defined by a mapping $\sigma: \Omega^* \to \Sigma$.
Namely, the agent observes the realization  $\om$ and corresponding utilities, and best-responds via $\sigma(\om_1,\ldots, \om_n)$. 
We often write $\sigma^*$ to denote the agent's best-response.

To design an effective mechanism, the principal uses a prior information regarding utility of the solution.
Formally, we mostly focus on \emph{agent-oblivious} \emph{Bayesian} mechanisms which know the distribution of $X_i$ but not $Y_i$ for $i \in [n]$, and whose performances are measured with respect to the worst-case agent's distribution of $Y_i$'s.
Correspondingly, we assume that the principal cannot observe the agent's utility for each element in agent-oblivious mechanisms.
We often consider \emph{agent-aware} mechanisms which know both distributions of $X_i$ and $Y_i$, for $i \in [n]$.

Essentially, \cite{kleinberg2018delegated} consider the following class of single proposal mechanisms instead of the general signaling mechanisms.
\begin{definition}[Single-proposal mechanism]
    In a \emph{single proposal mechanism} (SPM), the principal announces an eligible set $E \in \R_{\ge 0}^2$, the agent proposes an element $\om^* \in \Omega$, and the principal accepts the element if $(x(\om^*),y(\om^*)) \in E$, and rejects it otherwise.
\end{definition}
In short, the agent proposes a single solution and the principal decides whether or not to accept based on the eligible set.
Notably, they formally prove that it is without loss of generality to focus on single-proposal mechanisms due to a revelation-principle-style of argument.
\cite{kleinberg2018delegated} show that there exists a SPM which achieves a constant approximation to the first-best outcome in hindsight, $\Ex{\max_{i \in [n]} X_i}$, by exploring a novel connection to a version of the prophet inequality~\citep{krengel1987prophet}.

\begin{theorem}[\cite{kleinberg2018delegated}]\label{thm:kk18}
    There exists a SPM with a single-threshold that yields the principal's expected utility of at least $\frac{1}{2}\Ex{\max_{i \in [n]} X_i}$.
\end{theorem}

%% file: static-cost.tex
\subsection{Delegated Choice Problem with Inspection Cost}
Importantly, however, one might notice that the definition of SPM and results therein build upon a critical assumption such that if the agent proposes the solution $\om_i$, then he truthfully reports the corresponding utility $x(\om_i)$ and $y(\om_i)$ to the principal so that the principal can reliably believe the reported utility and decide her action depending on it.
For instance, if the agent knows that the principal would simply trust the utility reported by him, the agent may strategically misreport the utility of the proposed solution to induce the principal to accept the solution.\footnote{The previous literature precludes such scenario by assuming that such a untruthful report might be easily verified, thus the agent would truthfully report the utility given the threat of future punishment or to maintain his reputation.}

Therefore, without any procedure to verify the proposed solution's value, the principal may realize arbitrarily bad utility if she simply accepts the proposal relying on the reported utility.
Hence, we consider a scenario in which the principal can further \emph{inspect} the elements to verify that the elements correspond to correct utilities. Under this assumption that the principal can verify the proposed element's utility, it is natural to allow the principal to further inspect the other elements, or just directly inspect the elements by herself without any delegation.
Our model captures all these scenarios.

Now we formally define the problem of \emph{delegated choice with inspection cost} (DCIC).
As in the delegated choice problem presented above, there exist $n$ alternatives $\om_1, \ldots, \om_n$, and the principal and the agent have corresponding utility functions $x,y: \Omega \to \R_{\ge 0}$.
Analogously, we write $X_i$ and $Y_i$ to denote $x(\om_i)$ and $y(\om_i)$.
Each alternative $\om_i$ is equipped with an inspection cost $c_i \ge 0$ for $i \in [n]$.
In addition, if the principal decides to delegate to the agent, it incurs an exogenous cost of $c_{\Del} \ge 0$ to the principal.

Regardless of whether the principal delegates or not, the principal needs to design an inspection policy described as follows.
\begin{definition}[Inspection policy]
    An \emph{inspection policy} $P$ is defined by a mapping $P:\cH \to ([n] \times \cA) \cup \{\text{stop}\}$ where $\cH = ([n] \times \R_{\ge 0})^*$ denotes the set of all possible inspection histories and $\cA = \{\text{select}, \text{inspect}\}$ denotes the set of possible actions on solutions, and $\{\text{stop}\}$ denotes the possible action of stopping without selecting any solution.
    In words, given any history $H \in \cH$ that consists of inspected elements and corresponding observed utilities, the inspection policy $P$ determines the next action of whether to select an element $i \in [n]$ (possibly without inspecting it), inspect an element $i$, or stop without selecting any solution.
    Let $\cP$ be the set of all possible inspection policies.
\end{definition}

If the principal decides to delegate to the agent, the principal should further commit to a signaling mechanism, defined as follows.
\begin{definition}[Signaling mechanism]
    In the \emph{signaling mechanism}, given a set of signals $\Sigma$, the principal commits to a mapping $\Sig: \Sigma \to \cP$ that maps the agent's signal $\sigma \in \Sigma$ to an inspection policy $P \in \cP$.
\end{definition}

Finally, the principal's overall mechanism can be fully characterized by $\Mec = (\bI_{\Del}, P, \Sig)$ where $\bI_{\Del} \in \{0,1\}$ denotes the indicator variable of whether the principal decides to delegate ($1$) or not ($0$), $P \in \cP$ denotes the direct inspection policy that will be run if $\bI_{\Del} = 0$, and $\Sig$ is the signaling mechanism $\Sig: \Sigma \to \cP$ that will be run if $\bI_{\Del} = 1$.

In particular, in DCIC problem, we consider a \emph{modified} version of SPM as follows:
\begin{definition}[Single-proposal mechanism with inspection]
    In a single proposal mechanism with inspection (SPMI), the principal announces an eligible set $E \in \R_{\ge 0}^2$, and the agent proposes an element $\om^* \in \Omega$.
    If $\om^* \neq \perp$, then the principal inspects $\om^*$ and selects it if $(x(\om^*),y(\om^*)) \in E$ and otherwise does not select any element.
    If $\om^* = \perp$ then the principal does not inspect nor selects any.
\end{definition}

Then, the timing of the game can be described as follows.
\begin{enumerate}
    \item The principal commits to a mechanism $\Mec = (\bI_{\Del}, P, \Sig)$
    \item If $\bI_{\Del} = 0$, the principal runs the  inspection policy $P$ and the outcome is determined by $P$.
    \item Otherwise, the agent best-responds according to strategy $\sigma^*: \Omega^* \to \Sigma$ which maximizes his own utility given $\Sig$, and the principal runs the signaling mechanism $\Sig(\sigma^*((\omega_1,\ldots,\omega_n)))$ to inspect the solutions and determine the outcome.
\end{enumerate}
Given a principal's mechanism $\Mec$ and the agent's response $\sigma^*$ if $\bI_{\Del}=1$, let $\bA_i$ be the indicator variable denoting whether the element $i$ is selected by the principal, and $\bI_i$ be the one denoting whether the principal inspects $i$, for $i \in [n]$.
Then, one can observe that the mechanism's expected utility given the agent's response can be written as
\begin{align*}
    \Ex{\Mec} = \Ex{\parans{\sum_{i \in [n]} \bA_i X_i - \bI_i c_i} - \bI_{\Del} c_{\Del}},
\end{align*}
where we write $\Ex{\Mec}$ to denote the principal's expected utility given $\Mec$.



Recall that~\cite{kleinberg2018delegated} considers a mechanism that approximates the first-best outcome in hindsight, \ie $\Ex{\max_{i \in [n]} X_i}$.
In our setting, however, it is straightforward to see that it is impossible to approximate the first-best outcome in hindsight, formalized as follows.

\begin{proposition}[Inapproximability against the first-best outcome]\label{prop:inapx-firstbest}
    In the DCIC, there exists a problem instance such that no mechanism can obtain a utility of $\alpha$ times the first-best outcome in hindsight, \ie $\Ex{\max_{i \in [n]} X_i}$ for any $\alpha = \Omega(1/n)$.
\end{proposition}

Thus, we consider a second-best benchmark that competes with the optimal mechanism.
Let $\Opt$ be an optimal mechanism that the principal could have run given the distributions of $X_i$s, which maximizes expected utility given the worst-case agent.
Accordingly, we will consider a polynomial time computable mechanism that approximates $\Opt$.
We write $\Ex{\Opt}$ to denote the corresponding principal's expected utility
with respect to the worst-case agent's utility.
Recall that $D_X = \times_{i \in [n]}D_{X_i}$ and $D_Y = \times_{i \in [n]}D_{Y_i}$ denote the sets of all possible distributions for $(X_i)_{i \in[n]}$ and $(Y_i)_{i \in [n]}$, respectively.
The principal's expected utility, given mechanism $\Mec$ and distribution $D_X$, can be written as
\begin{align}
    \Exu{D_X}{\Mec} = \min_{D_Y \in \cD_Y}\Exu{D_X,D_Y}{\parans{\sum_{i \in [n]} \bA_i X_i - \bI_i c_i} - \bI_{\Del} c_{\Del}}.\label{def:obl}
\end{align}
We focus on a set of deterministic mechanisms that deterministically set $\bI_i$ and $\bI_{\Del}$ to one or zero.\footnote{Pandora's box problem and its variant typically has no advantage of randomizing the policy, but DCIC can have some advantages since the agent's best-response could be made different by randomization. Studying randomized mechanism remains a major open problem.}
Given an optimal mechanism $\Opt$, the mechanism $\Mec$ is an $\alpha$-approximation\footnote{Note that this differs from computing worst-case ratio over every possible $(D_X, D_Y) \in (\cD_X, \cD_Y)$.
This is mainly due to the fact that the principal is oblivious to the agent's mechanism, thus the natural objective one could consider is to optimize the expected utility over the worst-case agent's distribution.} if
\begin{align}
    \alpha \cdot \min_{D_X \in \cD_X}\frac{\Exu{D_X}{\Mec}}{\Exu{D_X}{\Opt}} \ge 1.\label{def:apx}
\end{align}


We finally impose a minor assumption such that if the agent is indifferent to the eventual outcome of two actions, \ie facing the same utility, then the agent behaves in favor of the principal to choose an action that induces a higher utility for the principal.
This is innocuous since it is always better to make the principal satisfactory for the agent's reputation effect unless it hurts the agent's utility.


\paragraph{Pandora's box with nonobligatory inspection}
Interestingly, we first observe that DCIC generalizes the Pandora's box problem with nonobligatory inspection (PNOI) by~\cite{doval2018whether}.


\begin{definition}[Pandora's box with nonobligatory inspection]
    In the \emph{Pandora's box problem with nonobligatory inspection} (PNOI), there are boxes indexed by $1,\ldots, n$ each of which is equipped with a random reward $X_i$ for $i \in [n]$.
    The searcher knows the distribution of the random variables in priori.
    Opening box $i$ incurs cost $c_i$ and reveals the realization $X_i$.
    Each box can only be selected after it is opened, and selecting $i$ yields the reward of $X_i$.
    Importantly, the searcher can also select the box without inspecting it.
    The searcher's objective is to maximize the expected payoff which is defined as the expected reward of the selected box minus the total inspection costs.
\end{definition}

Essentially, each box can be viewed as an element where the reward from the box is the principal's utility for that solution.
Then, without any delegation process, the principal can exactly solve the Pandora's box problem with nonobligatory inspection.
We write \emph{direct inspection policy} to denote the policy entirely based on the Pandora's box with nonobligatory inspection, without any delegation.
Thus, we often use the word box to denote an element for DCIC, and vice versa for PNOI.

Indeed, if the cost of delegation $c_{\Del}$ is far much larger than the maximal value of the support of the reward $X_i$'s, then the principal would trivially set $I_{\Del} = 0$.
This implies that we can reduce PNOI problem to an instance of DCIC.
Thus, since PNOI is recently proven to be NP-hard by~\cite{fu2023pandora}, it directly follows that DCIC is also NP-hard.
\begin{proposition}
    The decision version of DCIC is NP-hard.
\end{proposition}

\section{Main Results}\label{sec:results}
This section presents our main results on the $3$-approximate  mechanism.
Our results and intuitions therein are presented in a step by step manner, \ie we first present several examples to construct our intuition towards an efficient mechanism, and  present our main techniques and analysis.

\subsection{Warm-up}
Naturally, the first question one may ask is whether either of the pure delegation mechanisms, such as SPMI or the pure PNOI policy, would achieve a constant approximation ratio.
We here provide a formal explanation on why such naive approaches would necessarily fail.

We start with the example stating that if the principal only considers a direct inspection mechanism, i.e., always choosing $I_{\Del} = 0$, there exists a problem instance such that the approximation factor can be arbitrarily large.

\begin{proposition}[Direct inspection fails]\label{prop:no-del-fails}
    There exists a problem instance in the independent utility setting in which the optimal direct inspection policy cannot achieve a better than $\Omega(n)$-approximation.
\end{proposition}

Then, one might wonder whether it is possible to obtain a good approximation if instead of precluding the delegation and directly inspecting the solutions, the principal commits to SPMI that only inspects the single solution proposed by the agent.

\begin{proposition}[SPMI fails]\label{prop:spm-fails}
    There exists a problem instance of DCIC in which the SPMI cannot be $\alpha$-approximate for any $\alpha > 0$.
\end{proposition}

\subsection{$3$-approximation in costless delegation setting}
Thus, it is necessary to construct a mechanism that considers the possibility of delegation as well as further inspection of the elements or selection without inspection.
To simplify our argument, we start with a simple scenario in which the cost of delegation is zero, \ie $c_{\Del} = 0$.
In this case, one can always set $I_{\Del}$ to one with zero cost since it is always better to delegate to the agent to extract some useful information rather than not delegating, and any mechanism $\Mec = (0, P, \Sig)$ with $I_{\Del}=0$ can be simulated by choosing a mechanism $\Mec' = (1,P,\Sig')$ where $\Sig'$ maps every signal to $P$, \ie ignore the agent's signal.

To analyze an approximation factor of a mechanism, it is essential to obtain a proper upper bound on the benchmark $\Ex{\Opt}$.
However, it is nontrivial to quantify what the optimal mechanism would have extracted as additional information from the agent's signal.
Indeed, the following example shows that the value of information could be significant for certain types of mechanisms.
\begin{proposition}\label{prop:info}
    There exists an instance of the problem and a corresponding mechanism where the ratio of utility gained from uninspected solutions given any agent's signal using this mechanism, \ie $E[\sum_iX_i\bA_i\cdot\Ind{\bI_i = 0}]$ to the maximum expected utility of any solution $\max_{i\in[n]}E[X_i]$ is $\Omega(n)$.
\end{proposition}

This ratio shows the value of information gained from delegation, as without delegation the maximum utility that can be gained from uninspected solutions is $\max_{i\in[n]}E[X_i]$.

Towards the constant approximate mechanism, we first establish an upper bound for the utility obtained from an optimal mechanism. We say a solution $i$ is accepted without \emph{overinspection} if it is selected without inspecting any solution of the same or higher inspection cost than $c_i$. 
Defining indicator variables $\bU_i=\bigvee_{j:c_j \geq c_i} \bI_j$, we can say solution $i$ is accepted without overinspection if and only if $\bA_i(1-\bU_i)=1$.

Now, we prove the following lemma: \
\begin{lemma}\label{lm:ub-opt}
    For any mechanism, there exists an agent that can best-respond in a way such that the principal's utility from solutions accepted without overinspection will be at most $\max_{i\in[n]}{\Ex{X_i}}$, \ie
    \begin{align*}
        \sum_{i \in [n]}\Ex{X_i \bA_i \cdot \Ind{\bU_i = 0}} \le \max_{i \in [n]}\Ex{X_i}.
    \end{align*}
\end{lemma}
\begin{proof}
Assume WLOG that the solutions are numbered in order of increasing inspection costs, \ie $c_i\leq c_j$ for $i<j$, with ties broken arbitrarily. We consider an agent with deterministic utility values $y=(y_1,\ldots,y_n)$ such that $y_i > y_j$ when $i<j$.

For a given realization of the principal's utilities $x=(x_1,\ldots,x_n)$, let $\unsig(x)$ be the set of signals that would lead to a solution being accepted without overinspection, and for any signal $\sigma \in \unsig(x)$ let $acc(x,\sigma)$ be the solution accepted using that signal.
Now, we consider the following response by the agent given realization $x$. If $\unsig(x)$ is non-empty, the agent chooses a signal $\sigma \in \unsig(x)$ minimizing $acc(x,\sigma)$ and consequently maximizing the agent's utility. Otherwise, if $\unsig(x)$ is empty, the agent selects an arbitrary signal.

We argue that the principal's utility gained from solutions accepted without overinspection is higher for this response, compared to any best-response of the agent. For any realization $x$, if the agent's best-response chooses a signal in $\unsig(x)$ that leads to a solution being selected without overinspection, the same solution must be selected by our response, as any other solution being selected based on a signal in $\unsig(x)$ has lower utility for the agent. If the best-response does not select a solution without overinspection, then the corresponding utility $\sum_{i \in [n]}{X_i \bA_i \cdot \Ind{\bU_i = 0}}$ is simply $0$ for this realization since $\bU_i =1$ when $\bA_i=1$,
and our statement holds. Therefore, we can use this response to upper bound the principal's utility given any best-response of the agent.

Now, for this response, consider any realization $x=(x_1,\ldots,x_n)$ where solution $i$ is accepted without overinspection. We show that for any other realization $x'=(x_1,\ldots,x_{i-1},x_i',\ldots,x_n)$ that differs from $x$ only in $x_i$, we will also accept solution $i$ without overinspection. Clearly, the same signal can be used to select solution $i$ without overinspection, so $\unsig(x')$ is not empty, and a signal in $\unsig(x')$ is used by the agent. 
Then, for solution $i$ not to be selected, a solution $j$ with $y_j>y_i$ needs to be chosen without overinspection using a signal $\sigma' \in \unsig(x')$. Since $y_j>y_i$, we have $j<i$ and therefore $c_j\leq c_i$. This means that $\sigma'$ should not inspect solution $i$ for solution $j$ to be selected without overinspection. But then, since $x$ and $x'$ differ only in $x_i$, using the same signal given $x$ should also lead to solution $j$ being accepted without overinspection, which is a contradiction. So, for any such $x'$, solution $i$ will be accepted without overinspection. 

Since changing $x_i$ can not affect whether solution $i$ is selected without overinspection or not, we can conclude that $\bA_i\Ind{\bU_i = 0}$, which represents solution $i$ being selected without overinspection, can be written as a function of the $X_j$ variables except $X_i$. Therefore,  $\bA_i\Ind{\bU_i = 0}$ is independent of $X_i$. We can now use this to show that a utility of at most $\max_{i\in[n]}\Ex{X_i}$ can be gained from the solutions accepted without overinspection:

\begin{align*}
    \sum_{i \in [n]}\Ex{X_i \bA_i \cdot \Ind{\bU_i = 0}}&=\sum_{i \in [n]}\Ex{X_i}\Ex{\bA_i\cdot \Ind{\bU_i = 0}}\tag{Independence}\\
    &\leq \sum_{i \in [n]}\Ex{X_i}\Pr{\bA_i=1}\\
    &\leq \sum_{i\in[n]}\Pr{\bA_i=1}\max_{j \in [n]}\Ex{X_j} \tag{$\Ex{X_i} \leq \max_{j \in [n]}\Ex{X_j}$}\\
    &\leq \max_{i \in [n]}\Ex{X_i} \tag{$\sum_{i\in [n]}\Pr{\bA_i=1} \leq 1$},
\end{align*}
where the first inequality follows from the fact that $\Ind{\bU_i=0}\leq 1$ and $\Ex{\bA_i}=\Pr{\bA_i=1}$.
\end{proof}
Next, we use this lemma to prove an upper bound for the utility obtained by any mechanism. 
\begin{lemma} \label{lm:opt-fullub}
    For any mechanism, there exists an agent such that the utility obtained by the principal given the agent's best-response is at most
    $\max_{i\in[n]} \Ex{X_i} + \Ex{\max_{i\in[n]} (X_i - c_i)^+}$. That is,
    $$
    \Ex{\sum_{i\in[n]} X_i \bA_i - \max_{i\in[n]}c_i \bI_i} \leq \max_{i\in[n]} \Ex{X_i} + \Ex{\max_{i\in[n]} (X_i - c_i)^+}.
    $$
\end{lemma}
\begin{proof}
Let us denote the left-hand side of the inequality in the theorem  as $\ub$. We can use the fact that at most one $\bA_i$ can be $1$ to upper bound $\ub$ as follows:
$$
\ub = \Ex{\sum_{i\in[n]} X_i \bA_i - \max_{i\in[n]}c_i \bI_i} \leq
    \Ex{\sum_{i\in[n]} (X_i \bA_i - \bA_i (\max_{j\in[n]}c_j \bI_j))}.
$$
Now, we have a sum over $i$ and for each index $i$, the sum is only non-zero when $\bA_i$ is $1$ and solution $i$ is selected. We use the notion of selection without overinspection we introduced to break the sum down into two cases:
\begin{align*}
    \ub&=\Ex{\sum_{i\in[n]} (X_i \bA_i - \bA_i (\max_{j\in[n]}c_j \bI_j))}\\&=\Ex{\sum_{i\in[n]} (X_i \bA_i - \bA_i (\max_{j\in[n]}c_j \bI_j))(\bU_i + (1-\bU_i))}\\
    &\leq\Ex{\sum_{i\in[n]} (X_i \bA_i - \bA_i (\max_{j\in[n]}c_j \bI_j))\bU_i} + \Ex{\sum_{i\in[n]} X_i \bA_i\Ind{\bU_i=0}}.
\end{align*}    
In the second case, we disregard the inspection costs to obtain an upper bound. Next, in the first case, since the summation term is only non-zero when $\bU_i=1$, we can safely replace the maximum cost of inspected solutions with $c_i$. So, we can say
\begin{align*}
    \ub&\leq \Ex{\sum_{i\in[n]} (X_i \bA_i - \bA_i c_i)\bU_i} + \Ex{\sum_{i\in[n]} X_i \bA_i\Ind{\bU_i=0}}\\
    &\leq \Ex{\sum_{i\in[n]} \bA_i (X_i  - c_i)^+} + \Ex{\sum_{i\in[n]} X_i \bA_i\Ind{\bU_i=0}}\\
    &\leq \Ex{\max_{i\in[n]} (X_i  - c_i)^+} + \Ex{\sum_{i\in[n]} X_i \bA_i\Ind{\bU_i=0}}.
\end{align*}
Finally, we can use Lemma \ref{lm:ub-opt} to replace the second term and achieve the desired upper bound:
\begin{align*}
    \ub&\leq \Ex{\max_{i\in[n]} (X_i  - c_i)^+} + \max_{i\in[n]} \Ex{X_i}.
\end{align*}
\end{proof}

To obtain a $3$-approximate mechanism, we mainly consider a maximal mechanism that computes the expected value of two mechanisms and takes the one that has the higher value.
Those two mechanisms are fairly simple: (i) SPMI with an appropriate threshold, and (ii) that simply selects the one that has the largest expected utility, \ie selecting $\argmax_{i \in [n]}\Ex{X_i}$ without inspection.\footnote{The second mechanism will be replaced with PTAS policy for PNOI in Section~\ref{sec:general}}


We first prove a lemma stating that the the variant of the SPMI achieves half of $\Ex{\max_{i \in [n]} \parans{X_i - c_i}^+}$.
\begin{lemma}\label{lm:delegation}
    There exists a SPMI, which inspects the proposed element and accepts it based on a single threshold $\tau$, that achieves $\frac{1}{2}\Ex{\max_{i \in [n]} \parans{X_i - c_i}^+}$.\footnote{We will disregard the computation required to compute $\tau$ as it's beyond our interest. In fact, this can be efficiently computable for discrete random variables, and can be approximately computable for continuous random variables.}
\end{lemma}
\begin{proof}
    Given an instance of DCIC with $X = (X_1,\ldots, X_n)$, consider an instance of DCIC with $X' = (X_1',\ldots, X_n')$ such that the distributions are shifted to be $X_i' = \parans{X_i - c_i}^+$, but the inspection costs are set to be $c_i = 0$ for $i \in [n]$.
    Since the inspection is given at free, this coincides with the setting without inspection cost by~\cite{kleinberg2018delegated}.
    Thus, in this problem instance, by Lemma~\ref{thm:kk18}, we know that there exists a SPMI $\Mec'$ with a threshold $\tau$ that achieves $\Ex{\Mec'} \ge \frac{1}{2}\Ex{\max_{i \in [n]} \parans{X_i - c_i}^+}$.
    
    Now, in the original instance, consider the SPMI that accepts a solution $i$ after inspection if $X_i - c_i \ge \tau$.
    We claim that the principal's utility in $\Mec$ is the same as in $\Mec'$ for each realization of $X$ and $X'$ assuming that $X_i$ and $X_i'$ are coupled so that $X_i' = \parans{X_i - c_i}^+$ for every $i \in [n]$.
    If there exists an eligible solution $i$ such that $X_i-c_i \geq \tau$, we would also have $X_i'=X_i-c_i\geq \tau$ and the agent would propose his most preferred eligible solution $i$ in both cases, with the same utility of $X_i-c_i=X_i'$. Otherwise, if no eligible solution exists, the agent submits the signal $\perp$ in $\Mec$ resulting in a utility of $0$ for the principal since the agent's utility is $0$ regardless of the signal and he behaves in favor of the principal in case of ties. In $\Mec'$, the principal's utility will also be $0$ as no eligible solution exists.  
    Thus, we have $\Ex{\Mec} = \Ex{\Mec'} \ge \frac{1}{2}\Ex{\max_{i \in [n]} \parans{X_i - c_i}^+}$.
\end{proof}

Now, the following lemma shows that the maximal mechanism described above has the desired lower bound on the principal's utility.
\begin{lemma}\label{lm:simple}
    There exists a mechanism $\Mec$ that obtains a value of at least 
    $$\Ex{\Mec} \ge \max \left( \max_{i \in [n]}\Ex{X_i}, \frac{1}{2} \Ex{\max_{i \in [n]} (X_i - c_i)^+}\right).$$
\end{lemma}
\begin{proof}
    Consider a mechanism that simply selects an element that maximizes $\max_{i \in [n]}\Ex{X_i}$ without delegation nor inspection.
    
    Obviously the expected utility from this mechanism is $\max_{i \in [n]}\Ex{X_i}$ since the principal suffers no inspection cost nor delegation cost at all.
    Now, if $\max_{i \in [n]}\Ex{X_i}$, is larger than (or equal to) $\frac{1}{2} \Ex{\max_{i \in [n]} (X_i - c_i)^+}$, we run the mechanism above.
    Otherwise, we run the mechanism suggested by Lemma~\ref{lm:delegation}.
    Thus, this maximal mechanism achieves the maximum over two quantities.
\end{proof}
We are now ready to prove the following theorem.
\begin{theorem}\label{thm:costless-3pax}
    There exists a mechanism that achieves an approximation ratio of 3 for the delegated choice with inspection cost if the cost of delegation is zero. 
\end{theorem}
\begin{proof}
    By Lemma \ref{lm:opt-fullub}, we know that for an optimal mechanism \Opt, we have
    $$\Ex{\Opt} \le \max_{i \in [n]}\Ex{X_i} + \Ex{\max_{i \in [n]} (X_i - c_i)^+}.$$

    On the other hand, by Lemma~\ref{lm:simple}, we know that there exists a mechanism with
    \begin{align}
        \Ex{\Mec} \ge \max \left( \max_{i \in [n]}\Ex{X_i}, \frac{1}{2} \Ex{\max_{i \in [n]} (X_i - c_i)^+}\right).\label{eq:06271208}    
    \end{align}
    
    Thus, if $\frac{1}{2} \Ex{\max_{i \in [n]} (X_i - c_i)^+} \ge \max_{i \in [n]}\Ex{X_i}$, we obtain
    \begin{align*}
        \Ex{\Mec} \ge \frac{1}{2} \Ex{\max_{i \in [n]} (X_i - c_i)^+}, \Ex{\Opt} \le \frac{3}{2}\Ex{\max_{i \in [n]} (X_i - c_i)^+},
    \end{align*}
    which implies the $3$-approximation ratio.
    Otherwise, we have
    \begin{align*}
        \Ex{\Mec} \ge \max_{i \in [n]}\Ex{X_i}, \Ex{\Opt} \le 3\max_{i \in [n]}\Ex{X_i}
    \end{align*}
    which also implies the $3$-approximation ratio.
    
    

    
\end{proof}


Interestingly, we observe that the analysis carries over to the more complicated setting in which the cost function is rather combinatorial, not additive.

\begin{corollary}\label{cor:combinatorial}
    Consider a monotone combinatorial cost setting under which the cost of inspecting a set of solutions is given by a monotone set function $c: 2^{[n]} \to \R_{\ge 0}$. If the delegation is costless, then there exists a $3$-approximate mechanism.
\end{corollary}
\begin{proof}
        Let $c_i = c(\{i\})$ be the cost of inspecting element $i$ by itself. Lemma \ref{lm:ub-opt} is not dependent on the cost function used for the principal, so it can still be applied when considering costs $c_i$. Consequently, Lemma \ref{lm:opt-fullub} can be used in this case, to show that
        $$
        \Ex{\sum_{i\in[n]} X_i \bA_i - \max_{i\in[n]}c_i \bI_i} \leq \max_{i\in[n]} \Ex{X_i} + \Ex{\max_{i\in[n]} (X_i - c_i)^+}.
        $$
        Now, we can see that for any monotone function $c$, the principal's utility can be upper bounded as
        \begin{align*}
            \Ex{(\sum_{i\in[n]} X_i\bA_i) - c(\{j\mid \bI_j=1\})} &\leq \Ex{\sum_{i\in[n]} X_i\bA_i - \max_{\bI_j=1}c(\{j\})}\\
            &=\Ex{\sum_{i\in[n]} X_i\bA_i - \max_{i\in[n]} c_i\bI_i}.
        \end{align*}
        Therefore, the bound in Lemma \ref{lm:opt-fullub} can be used as an upper bound for the principal's utility.

        On the other hand, our algorithm when delegation is costless uses at most one inspection, so its cost is unchanged compared to the setting with additive costs $c_i$. Therefore, the same algorithm achieves a 3-approximation in this setting too. 
\end{proof}
Additionally, the following theorem shows an improved approximation ratio of $2$ when all inspection costs are the same, whose proof follows from considering the mechanism in Theorem~\ref{thm:main}.
\begin{proposition}\label{prop:identical}
    If $c_i = c_j$ for any $i, j \in [n]$, then there exists a $2$-approximate mechanism for DCIC with costless delegation.
\end{proposition}

Finally, the following theorem complements the above results by proving that the approximation ratio of 3 is tight for our algorithm for costless delegation. 
\begin{theorem}\label{thm:tightness}
    There exists a problem instance of DCIC such that the mechanism described in Theorem~\ref{thm:costless-3pax} cannot have approximation factor better than $3+\eps$ for any $\eps > 0$.
\end{theorem}
\begin{proof}
    For a small value of $\eps>0$, consider the following instance with 3 solutions: The first two solutions have identical distributions, where $X_i=\nicefrac{1}{\eps}$ with probability $\eps$ and $X_i=0$ with probability $1-\eps$. The inspection cost for these solutions is $0$. The third solution has a deterministic utility of $1$, but its inspection cost is also $1$. One optimal algorithm in this instance inspects the first two solutions in order, selecting them if the utility is non-zero. If both solutions report a utility of zero, the third box is chosen without inspection. The expected utility of this algorithm is
    $$
    \eps\cdot\frac{1}{\eps} + (1-\eps)\cdot\eps\cdot\frac{1}{\eps} + (1-\eps)^2=3-3\eps+\eps^2.
    $$
    
    Now, our algorithm considers the two options of delegation or choosing the solution with maximum expected utility without inspection. For delegation, we look at the value   
    $$\frac{1}{2}\Ex{\max_{i\in[n]}{X_i-c_i}}=\frac{1}{2}\left(\eps\cdot\frac{1}{\eps} + (1-\eps)\cdot\eps\cdot\frac{1}{\eps}\right)=1-\frac{\eps}{2}.$$ 
    The solution with maximum expected utility has a utility of $1$ as all solutions have an expected utility of $1$. Since the maximum expected utility is higher than the utility bound for delegation, our algorithm will choose the solution with the maximum expected utility without inspection, achieving a utility of $1$. So, the approximation ratio of our algorithm cannot be better than $3-3\eps+\eps^2$. As $\eps$ can be arbitrarily small, the approximation ratio of $3$ is tight for our algorithm.
\end{proof}

\subsection{Constant factor approximation for general setting}\label{sec:general}
We finally present the following theorem in the costly delegation setting.
\begin{theorem}\label{thm:main}
    Let $\Mec_P$ be the optimal mechanism for the PNOI problem (without delegation).
    \begin{enumerate}
        \item If $c_{\Del} = \alpha \Ex{\max_{i \in[n]} (X_i -c_i)^+}$ for some $\alpha < 1/2$, then there exists a $(3-4\alpha)/(1-2\alpha)$ approximate mechanism.
        \item If $c_{\Del} \ge \Ex{\max_{i \in [n]} (X_i - c_i)^+} - 2\Ex{\Mec_P}$, then there exists a $2-\eps$ approximate mechanism for any $\eps > 0$.
    \end{enumerate}
\end{theorem}
\begin{proof}
    Consider a PTAS mechanism $P_\eps$ for the PNOI problem by~\cite{fu2023pandora} and~\cite{beyhaghi2023pandora} that $1-\eps$ approximates $\Mec_P$.
    
    We consider the following mechanism $\Mec$.
    \begin{enumerate}
        \item Compute  $v_1 = \Ex{\Mec_{P_\eps}}$ of running $P_\eps$ without any delegation.
        \item Compute $v_2 = \frac{1}{2}\Ex{\max_{i \in [n]} (X_i - c_i)^+ }- c_{\Del}$ for running SPMI given by Lemma~\ref{lm:delegation}.
        \item If $v_1 \geq v_2$, run the mechanism $\Mec_P$, and otherwise run the above SPMI.
    \end{enumerate}
    Due to the construction, our mechanism satisfies
    \begin{align}
        \Ex{\Mec} &\ge \max\left(\frac{1}{2}\Ex{\max_{i \in [n]} (X_i - c_i)^+}-c_{\Del}, \Ex{\Mec_{P_\eps}}\right)\label{eq:costly-mec}
    \end{align}

    If the optimal mechanism $\Opt$, does not delegate, we have that
    \begin{align*}
        (1-\eps) \Ex{\Opt} \le \Ex{\Mec_P} \le \Ex{\Mec},
    \end{align*}
    thus our mechanism is $1-\eps$ approximation.

    On the other hand, if the optimal mechanism chooses to delegate, we can subtract $c_{\Del}$ from the bound in Lemma \ref{lm:opt-fullub} to say
    \begin{align*}
        \Ex{\Opt} &\le \max_{i \in [n]}\Ex{X_i} + \Ex{\max_{i \in [n]} (X_i - c_i)^+} - c_\Del\\
        &= \max_{i \in [n]}\Ex{X_i} + \Ex{\max_{i \in [n]} (X_i - c_i)^+ - c_\Del}.
    \end{align*}

    It is obvious that $v_1 \ge \max_{i \in [n]}\Ex{X_i}$.\footnote{One can also consider a maximal mechanism that mixes any mechanism with closed box inspection in the beginning.}
    Now we consider the first scenario where $c_\Del =\alpha \Ex{\max_{i \in [n]} (X_i -c_i)^+}$.
    Then, from~\eqref{eq:costly-mec} we have
    \begin{align*}
        \Ex{\Mec} \ge \max\parans{(\frac{1}{2}-\alpha)\Ex{\max_{i \in [n]} (X_i -c_i)^+}, \Ex{\Mec_{P_\eps}}}.
    \end{align*}
    We have the following two cases:
    \begin{enumerate}
        \item $\Ex{\Mec_{P_\eps}} \le \frac{1}{2}\Ex{\max_{i \in [n]} (X_i -c_i)^+} - c_{\Del}\leq(\frac{1}{2}-\alpha)\Ex{\max_{i \in [n]} (X_i -c_i)^+}$.
        \item $\Ex{\Mec_{P_\eps}} > \frac{1}{2}\Ex{\max_{i \in [n]} (X_i -c_i)^+} - c_{\Del}>(\frac{1}{2}-\alpha)\Ex{\max_{i \in [n]} (X_i -c_i)^+}$.
    \end{enumerate}
    In the first case, we have
    \begin{align*}
        \Ex{\Opt} &\leq \max_{i \in [n]}\Ex{X_i} + \Ex{\max_{i \in [n]} (X_i - c_i)^+ - c_\Del}\\
        &\le \Ex{\Mec_{P_\eps}} + \Ex{\max_{i \in [n]} (X_i - c_i)^+ - c_\Del}\\
        &\le (\frac{1}{2}-\alpha)\Ex{\max_{i \in [n]} (X_i - c_i)^+} + \Ex{\max_{i \in [n]} (X_i - c_i)^+ - c_\Del}\\
        &\le (\frac{3}{2}-2\alpha) \Ex{\max_{i \in [n]} (X_i - c_i)^+}
        \\
        &\le \frac{\frac{3}{2}-2\alpha}{1/2-\alpha} \Ex{\Mec},
    \end{align*}
    implying that $\Mec$ is $(3-4\alpha)/({1-2\alpha})$-approximation.
    \\

    Otherwise, we have
    \begin{align*}
        \Ex{\Opt} &\leq \max_{i \in [n]}\Ex{X_i} + \Ex{\max_{i \in [n]} (X_i - c_i)^+ - c_\Del}\\
        &\le \Ex{\Mec_{P_\eps}} + \Ex{\max_{i \in [n]} (X_i - c_i)^+ - c_\Del}\\
        &\le \Ex{\Mec_{P_\eps}} + (1-\alpha)\Ex{\max_{i \in [n]} (X_i - c_i)^+}\\
        &\le (1 + \frac{1-\alpha}{\nicefrac{1}{2}-\alpha}) \Ex{\Mec_{P_\eps}}
        \\
        &\le \frac{\nicefrac{3}{2}-2\alpha}{\nicefrac{1}{2}-\alpha} \Ex{\Mec},
    \end{align*}
    implying that $\Mec$ is a $(3-4\alpha)/({1-2\alpha})$-approximation.
    Since in both cases $\Mec$ is a $(3-4\alpha)/({1-2\alpha})$ approximation, the mechanism is a $(3-4\alpha)/({1-2\alpha})$ approximation in the scenario where $c_{\Del}=\alpha\Ex{\max_{i \in[n]} (X_i -c_i)^+}$

    Now consider the second case where $c_{\Del} \ge \Ex{\max_{i \in [n]}(X_i - c_i)^+} - \Ex{\Mec_P}$.
    Then, we have
    \begin{align*}
        \Ex{\Mec_{P_\eps}} 
        \ge (1-\eps)\Ex{\Mec_P} 
        \ge (1-\eps)\parans{\Ex{\max_{i \in [n]}(X_i - c_i)^+} - c_{\Del}}.
    \end{align*}
    In this case, we have
    \begin{align*}
        \Ex{\Opt} 
        &\le \Ex{\Mec_{P_\eps}} + \Ex{\max_{i\in [n]} (X_i - c_i)^+ - c_{\Del}}
        \\
        &\le \parans{1 + \frac{1}{1-\eps}}\Ex{\Mec_{P_\eps}} 
        \\
        &\le \parans{1 + \frac{1}{1-\eps}}\Ex{\Mec},
    \end{align*}
    concluding that it is $(2-\eps')$-approximation for any $\eps' > 0$ by taking $\eps$ sufficiently small.

    
\end{proof}

\section{Concluding Remarks}
Our work inspires several interesting directions to explore. The most immediate direction would be to obtain improved approximate mechanism, possibly PTAS, in general setting.
We conjecture that DCIC might have strictly harder computational complexity than PNOI as the intrinsic difficulties in analyzing every class of signaling mechanisms imparts more technical challenges in the problem. 
Since randomized mechanism may change the agent's best-response, studying the impact of randomization would be also interesting.
It would be also interesting to investigate a mechanism that knows the agent's distribution a priori. 
Several other variants of the Pandora's box problem, \eg combinatorial valuation by~\cite{singla2018price}, when the recall is not possible~\citep{kleinberg2016descending}, and sample based mechanism or correlated setting by~\cite{chawla2020pandora} could analogously be studied in our setting when the side information from the strategic agent is possible.



%% file: ec2025/appendix.tex
\section{Remaining Proofs}
\subsection{Proof of Proposition~\ref{prop:inapx-firstbest}}
\begin{proof}[Proof of Proposition~\ref{prop:inapx-firstbest}]
    Consider $n$ ex-ante identical elements such that $X_i$'s are i.i.d. random variables from a distribution $D$ such that $X_i = 1/\eps$ with probability $\eps$ and otherwise $0$.
    Let each element $i$'s inspection cost be $c_i = c = (1- \eps)\cdot \Ex{\max_{i\in[n]}{X_i}}$.
    If the principal's mechanism inspects at least one element, her expected utility will be
    \begin{align*}
        \Ex{\Mec} = \Ex{\parans{\sum_{i \in [n]} \bA_i X_i - \bI_i c_i} - \bI_{\Del} c_{\Del}}
        \le
        \Ex{\max_{i \in [n]} X_i} - (1-\eps) \Ex{\max_{i \in [n]} X_i} 
        &= \eps \Ex{\max_{i \in [n]}X_i}
    \end{align*}
    On the other hand, the principal may select an element without inspecting any element.
    In this case, the utility will be at most
    \begin{align*}
        \Ex{\Mec} \le \max_{i \in [n]}\Ex{X_i} = \Ex{X_i} = 1.
    \end{align*}
    Note, however, that the first-best outcome in hindsight is $\Ex{\max_{i \in [n]} X_i} = (1-(1-\eps)^n) \cdot 1/\eps$.
    By plugging $\eps = 1/n$, we obtain
    \begin{align*}
        \Ex{\max_{i \in [n]} X_i} = (1-(1-1/n)^n) \cdot n \ge \left(1-\frac{1}{e}\right)n,
    \end{align*}
    where we use $(1-1/n)^n \le 1/e$.
    On the other hand, we know that
    \begin{align*}
        \Ex{\Mec} \le \max\left(1, \eps \Ex{\max_{i \in [n]} X_i}\right) \le  \max\left(1, \frac{1}{n} \cdot n\right) = 1,
    \end{align*}
    which completes the proof.
\end{proof}

\subsection{Proof of Proposition~\ref{prop:no-del-fails}}
\begin{proof}[Proof of Proposition~\ref{prop:no-del-fails}]
    Suppose we have $n$ identical solutions with the same utility distribution for the principal, where 
    each solution has utility $1$ with probability $p$ and $0$ otherwise.
    Additionally, suppose that the cost of inspecting each solution is $c$.

    The probability that at least one solution is realized to have utility $X_i=1$ is
    \begin{align*}
        \Pr{\max_{i\in[n]}X_i=1} = 1-(1-p)^n .
    \end{align*}

    Now, consider a SPMI with an eligibility threshold of $1/2 > 0$.
    Any rational agent will report a solution with non-zero utility if one exists since otherwise the utility will simply be zero.
    Therefore, the principal's expected utility for the delegation-only mechanism is
    \begin{align*}
        \Ex{\Mec} \ge 1-(1-p)^n - c \; .
    \end{align*}
    Thus $\Ex{\Opt} \ge \Ex{\Mec} = 1-(1-p)^n - c$.

    On the other hand, we consider a direct inspection policy.
    Since all the solutions have identical utilities for the principal along with identical costs, and since there is only one non-zero possible utility value, the number of solutions inspected is enough to describe all possible inspection policies. 
    So, for any number $0 \leq k \leq n$, we consider the expected utility of a policy that inspects up to $k$ solutions, selecting the first one that achieves a non-zero utility. If no inspected solution's utility is non-zero and $k<n$, the policy selects an uninspected solution with an expected utility of $p$.
    Now, we can provide an upper bound for the utility of such a mechanism. If $k \geq 1$, we can say
    \begin{align*}
        \Ex{\Mec} &\leq \sum_{i=1}^{k} p(1-p)^{i-1}(1 - ic) + (1-p)^k (p - kc)\\
        &=\sum_{i=1}^{k} p(1-p)^{i-1} - \sum_{i=1}^{k}(1-p)^{i-1} c + p(1-p)^k\\
        &=p\frac{1-(1-p)^k}{1-(1-p)} - c\frac{1-(1-p)^k}{1-(1-p)} + p(1-p)^k\\
        &=(1-(1-p)^k)(1-\frac{c}{p}) + p(1-p)^k\\
        &\leq (1-(1-p)^k)(1-\frac{c}{p}) + p
    \end{align*}
    Now, let $p=\nicefrac{1}{n}$, $c=\nicefrac{2}{n}$,   and assume $n\geq 6$. Then $1-\nicefrac{c}{p} = 1 - 2 \le 0$, so we have for any inspection only mechanism $\Mec$
    $$\Ex{\Mec} \leq p = \frac{1}{n}.$$
    On the other hand, we have
    $$\Ex{\Opt} \geq (1-(1-\frac{1}{n})^n) - \frac{2}{n} \geq 1 - 1/e - \frac{2}{n} \geq \frac{1}{6}.$$
    Therefore, no mechanism utilizing only inspection can achieve an approximation ratio better than $\nicefrac{n}{6}=\Omega(n)$.
    
\end{proof}

\subsection{Proof of Proposition~\ref{prop:spm-fails}}
\begin{proof}[Proof of Proposition~\ref{prop:spm-fails}]
    Consider any instance where the cost of each solution corresponds to its possible maximum utility for the principal. For example, we can assume an instance with solutions that have $X_i=1$ with probability $1/2$ and $0$ with probability $1/2$, and each solution's inspection cost is $1$. It is clear that any SPMI that always inspects the nonnull solution proposed will have a utility of at most $0$, as the inspection cost is no less than the utility of the solution. On the other hand, choosing a solution without inspection leads to a positive utility of $1/2$. So, any $\alpha$-approximation is impossible using only this class of mechanisms. 
\end{proof}

\subsection{Proof of Proposition~\ref{prop:info}}
\begin{proof}[Proof of Proposition~\ref{prop:info}]
    Consider a problem instance with $n$ identical solutions where each $X_i$ is $1$ with probability $\eps$ and $0$ otherwise. 
    Then, the expected utility of each solution is $\eps$.
    Now, consider a signaling mechanism with $\Sigma=[n]$, with each signal corresponding to a solution. For signal $i$, the principal inspects all solutions except solution $i$, and accepts solution $i$ if all other solutions have a utility of $0$, and does not accept any solution otherwise.
    Then, any rational agent will use signal $i$ if solution $i$ is the only solution with a non-zero utility for the principal.
    It can be seen that assuming inspection costs are equal to $0$, the utility gained from uninspected solutions in this scenario will be $n\eps(1-\eps)^{n-1}$. 
    Then, the ratio of utility from uninspected solutions to the maximum expected utility of any solution can be written as $\frac{n\eps(1-\eps)^{n-1}}{\eps}=n(1-\eps)^{n-1}$ which tends to $n$ as $\eps$ goes toward $0$. 
\end{proof}
\subsection{Proof of Proposition~\ref{prop:identical}}
\begin{proof}[Proof of Proposition~\ref{prop:identical}]
    Consider the optimal mechanism $\Opt=(\bI_\Del,P,\Sig)$ and our mechanism from Theorem \ref{thm:costless-3pax}. If for any signal $\sigma \in \Sigma$, inspection policy $\Sig(\sigma)$ selects a solution $i$ without inspecting any solution, an agent preferring $i$ to other solutions will always use this signal. Then, the principal's utility can be upper bounded by $\Ex{X_i} \leq \max_{j\in[n]} \Ex{X_j}$. 
    
    Otherwise, for any signal $\sigma$, the policy $\Sig(\sigma)$ inspects at least one solution at cost $c$. Therefore, the expected utility obtained by the principal can be upper bounded by $\Ex{\max_{i\in[n]} X_i - c}$. By Lemma \ref{lm:simple}, our mechanism achieves a utility of at least $\max(\max_{i\in[n]}\Ex{X_i},\frac{1}{2}\Ex{\max_{i\in[n]} X_i - c})$, which is at least a 2-approximation to the optimal in each case.
\end{proof}

%% file: main.bbl
\begin{thebibliography}{35}
\providecommand{\natexlab}[1]{#1}
\providecommand{\url}[1]{\texttt{#1}}
\expandafter\ifx\csname urlstyle\endcsname\relax
  \providecommand{\doi}[1]{doi: #1}\else
  \providecommand{\doi}{doi: \begingroup \urlstyle{rm}\Url}\fi

\bibitem[Alonso and Matouschek(2008)]{alonso2008optimal}
Ricardo Alonso and Niko Matouschek.
\newblock Optimal delegation.
\newblock \emph{The Review of Economic Studies}, 75\penalty0 (1):\penalty0
  259--293, 2008.

\bibitem[Armstrong and Vickers(2010)]{armstrong2010model}
Mark Armstrong and John Vickers.
\newblock A model of delegated project choice.
\newblock \emph{Econometrica}, 78\penalty0 (1):\penalty0 213--244, 2010.

\bibitem[Bechtel and Dughmi(2020)]{bechtel2020delegated}
Curtis Bechtel and Shaddin Dughmi.
\newblock Delegated stochastic probing.
\newblock \emph{arXiv preprint arXiv:2010.14718}, 2020.

\bibitem[Bechtel et~al.(2022)Bechtel, Dughmi, and Patel]{bechtel2022delegated}
Curtis Bechtel, Shaddin Dughmi, and Neel Patel.
\newblock Delegated pandora's box.
\newblock In \emph{Proceedings of the 23rd ACM Conference on Economics and
  Computation}, pages 666--693, 2022.

\bibitem[Berger et~al.(2023)Berger, Ezra, Feldman, and
  Fusco]{berger2023pandora}
Ben Berger, Tomer Ezra, Michal Feldman, and Federico Fusco.
\newblock Pandora's problem with combinatorial cost.
\newblock In \emph{Proceedings of the 24th ACM Conference on Economics and
  Computation}, pages 273--292, 2023.

\bibitem[Beyhaghi and Cai(2023)]{beyhaghi2023pandora}
Hedyeh Beyhaghi and Linda Cai.
\newblock Pandora’s problem with nonobligatory inspection: Optimal structure
  and a ptas.
\newblock In \emph{Proceedings of the 55th Annual ACM Symposium on Theory of
  Computing}, pages 803--816, 2023.

\bibitem[Beyhaghi and Kleinberg(2019)]{beyhaghi2019pandora}
Hedyeh Beyhaghi and Robert Kleinberg.
\newblock Pandora's problem with nonobligatory inspection.
\newblock In \emph{Proceedings of the 2019 ACM Conference on Economics and
  Computation}, pages 131--132, 2019.

\bibitem[Bowers and Waggoner(2024)]{bowers2024matching}
Robin Bowers and Bo~Waggoner.
\newblock Matching with nested and bundled pandora boxes.
\newblock \emph{arXiv preprint arXiv:2406.08711}, 2024.

\bibitem[Chawla et~al.(2020)Chawla, Gergatsouli, Teng, Tzamos, and
  Zhang]{chawla2020pandora}
Shuchi Chawla, Evangelia Gergatsouli, Yifeng Teng, Christos Tzamos, and Ruimin
  Zhang.
\newblock Pandora's box with correlations: Learning and approximation.
\newblock In \emph{2020 IEEE 61st Annual Symposium on Foundations of Computer
  Science (FOCS)}, pages 1214--1225. IEEE, 2020.

\bibitem[{\v{C}}opi{\v{c}} and Ponsat{\'\i}(2008)]{vcopivc2008robust}
Jernej {\v{C}}opi{\v{c}} and Clara Ponsat{\'\i}.
\newblock Robust bilateral trade and mediated bargaining.
\newblock \emph{Journal of the European Economic Association}, 6\penalty0
  (2-3):\penalty0 570--580, 2008.

\bibitem[Dai et~al.(2024)Dai, Flanigan, Jagadeesan, Haghtalab, and
  Podimata]{dai2024can}
Jessica Dai, Bailey Flanigan, Meena Jagadeesan, Nika Haghtalab, and Chara
  Podimata.
\newblock Can probabilistic feedback drive user impacts in online platforms?
\newblock In \emph{International Conference on Artificial Intelligence and
  Statistics}, pages 2512--2520. PMLR, 2024.

\bibitem[Daskalakis et~al.(2013)Daskalakis, Deckelbaum, and
  Tzamos]{daskalakis2013mechanism}
Constantinos Daskalakis, Alan Deckelbaum, and Christos Tzamos.
\newblock Mechanism design via optimal transport.
\newblock In \emph{Proceedings of the fourteenth ACM conference on Electronic
  commerce}, pages 269--286, 2013.

\bibitem[Doval(2018)]{doval2018whether}
Laura Doval.
\newblock Whether or not to open pandora's box.
\newblock \emph{Journal of Economic Theory}, 175:\penalty0 127--158, 2018.

\bibitem[Eilat and Pauzner(2021)]{eilat2021bilateral}
Ran Eilat and Ady Pauzner.
\newblock Bilateral trade with a benevolent intermediary.
\newblock \emph{Theoretical Economics}, 16\penalty0 (4):\penalty0 1655--1714,
  2021.

\bibitem[Fu et~al.(2018)Fu, Li, and Xu]{fu2018ptas}
H~Fu, J~Li, and P~Xu.
\newblock A ptas for a class of stochastic dynamic programs. chatzigiannakis i,
  kaklamanis c, marx d, sannella d, eds.
\newblock In \emph{Proc. 45th Internat. Colloquium on Automata, Languages, and
  Programming.(ICALP 2018), Prague, Czech Republic}, pages 1--56, 2018.

\bibitem[Fu et~al.(2023)Fu, Li, and Liu]{fu2023pandora}
Hu~Fu, Jiawei Li, and Daogao Liu.
\newblock Pandora box problem with nonobligatory inspection: Hardness and
  approximation scheme.
\newblock In \emph{Proceedings of the 55th Annual ACM Symposium on Theory of
  Computing}, pages 789--802, 2023.

\bibitem[Guha et~al.(2008)Guha, Munagala, and Sarkar]{guha2008information}
Sudipto Guha, Kamesh Munagala, and Saswati Sarkar.
\newblock Information acquisition and exploitation in multichannel wireless
  networks.
\newblock \emph{arXiv preprint arXiv:0804.1724}, 2008.

\bibitem[Hajiaghayi et~al.(2024{\natexlab{a}})Hajiaghayi, Hajiaghayi, Peng, and
  Shin]{hajiaghayi2024gains}
Ilya Hajiaghayi, MohammadTaghi Hajiaghayi, Gary Peng, and Suho Shin.
\newblock Gains-from-trade in bilateral trade with a broker.
\newblock \emph{arXiv preprint arXiv:2410.17444}, 2024{\natexlab{a}}.

\bibitem[Hajiaghayi et~al.(2024{\natexlab{b}})Hajiaghayi, Mahdavi, Rezaei, and
  Shin]{hajiaghayi2024regret}
Mohammad Hajiaghayi, Mohammad Mahdavi, Keivan Rezaei, and Suho Shin.
\newblock Regret analysis of repeated delegated choice.
\newblock In \emph{Proceedings of the AAAI Conference on Artificial
  Intelligence}, volume~38, pages 9757--9764, 2024{\natexlab{b}}.

\bibitem[Holmstrom(1980)]{holmstrom1980theory}
Bengt Holmstrom.
\newblock On the theory of delegation.
\newblock Technical report, Discussion Paper, 1980.

\bibitem[Immorlica et~al.(2024)Immorlica, Lucier, and
  Slivkins]{immorlica2024generative}
Nicole Immorlica, Brendan Lucier, and Aleksandrs Slivkins.
\newblock Generative ai as economic agents.
\newblock \emph{ACM SIGecom Exchanges}, 22\penalty0 (1):\penalty0 93--109,
  2024.

\bibitem[Jagadeesan et~al.(2024)Jagadeesan, Garg, and
  Steinhardt]{jagadeesan2024supply}
Meena Jagadeesan, Nikhil Garg, and Jacob Steinhardt.
\newblock Supply-side equilibria in recommender systems.
\newblock \emph{Advances in Neural Information Processing Systems}, 36, 2024.

\bibitem[Khodabakhsh et~al.(2024)Khodabakhsh, Pountourakis, and
  Taggart]{khodabakhsh2024simple}
Ali Khodabakhsh, Emmanouil Pountourakis, and Samuel Taggart.
\newblock Simple delegated choice.
\newblock In \emph{Proceedings of the 2024 Annual ACM-SIAM Symposium on
  Discrete Algorithms (SODA)}, pages 569--590. SIAM, 2024.

\bibitem[Kleinberg and Kleinberg(2018)]{kleinberg2018delegated}
Jon Kleinberg and Robert Kleinberg.
\newblock Delegated search approximates efficient search.
\newblock In \emph{Proceedings of the 2018 ACM Conference on Economics and
  Computation}, pages 287--302, 2018.

\bibitem[Kleinberg et~al.(2016)Kleinberg, Waggoner, and
  Weyl]{kleinberg2016descending}
Robert Kleinberg, Bo~Waggoner, and E~Glen Weyl.
\newblock Descending price optimally coordinates search.
\newblock \emph{arXiv preprint arXiv:1603.07682}, 2016.

\bibitem[Krengel and Sucheston(1987)]{krengel1987prophet}
Ulrich Krengel and Louis Sucheston.
\newblock Prophet compared to gambler: an inequality for transforms of
  processes.
\newblock \emph{The Annals of Probability}, 15\penalty0 (4):\penalty0
  1593--1599, 1987.

\bibitem[Kuang et~al.(2023)Kuang, Shen, and Wu]{kuang2023profit}
Zhonghong Kuang, Weiran Shen, and Fan Wu.
\newblock Profit-maximizing mechanism in bilateral trade with interdependent
  valuations.
\newblock \emph{Available at SSRN 4474002}, 2023.

\bibitem[Li and Yao(2013)]{li2013revenue}
Xinye Li and Andrew Chi-Chih Yao.
\newblock On revenue maximization for selling multiple independently
  distributed items.
\newblock \emph{Proceedings of the National Academy of Sciences}, 110\penalty0
  (28):\penalty0 11232--11237, 2013.

\bibitem[Manelli and Vincent(2007)]{manelli2007multidimensional}
Alejandro~M Manelli and Daniel~R Vincent.
\newblock Multidimensional mechanism design: Revenue maximization and the
  multiple-good monopoly.
\newblock \emph{Journal of Economic theory}, 137\penalty0 (1):\penalty0
  153--185, 2007.

\bibitem[Samuel-Cahn(1984)]{samuel1984comparison}
Ester Samuel-Cahn.
\newblock Comparison of threshold stop rules and maximum for independent
  nonnegative random variables.
\newblock \emph{the Annals of Probability}, pages 1213--1216, 1984.

\bibitem[Shin et~al.(2022)Shin, Lee, and Ok]{shin2022multi}
Suho Shin, Seungjoon Lee, and Jungseul Ok.
\newblock Multi-armed bandit algorithm against strategic replication.
\newblock In \emph{International Conference on Artificial Intelligence and
  Statistics}, pages 403--431. PMLR, 2022.

\bibitem[Shin et~al.(2023)Shin, Rezaei, and Hajiaghayi]{shin2023delegating}
Suho Shin, Keivan Rezaei, and Mohammadtaghi Hajiaghayi.
\newblock Delegating to multiple agents.
\newblock In \emph{Proceedings of the 24th ACM Conference on Economics and
  Computation}, pages 1081--1126, 2023.

\bibitem[Singla(2018)]{singla2018price}
Sahil Singla.
\newblock The price of information in combinatorial optimization.
\newblock In \emph{Proceedings of the twenty-ninth annual ACM-SIAM symposium on
  discrete algorithms}, pages 2523--2532. SIAM, 2018.

\bibitem[Weitzman(1978)]{weitzman1978optimal}
Martin Weitzman.
\newblock \emph{Optimal search for the best alternative}, volume~78.
\newblock Department of Energy, 1978.

\bibitem[Yao et~al.(2024)Yao, Li, Sankararaman, Liao, Zhu, Wang, Wang, and
  Xu]{yao2024rethinking}
Fan Yao, Chuanhao Li, Karthik~Abinav Sankararaman, Yiming Liao, Yan Zhu, Qifan
  Wang, Hongning Wang, and Haifeng Xu.
\newblock Rethinking incentives in recommender systems: are monotone rewards
  always beneficial?
\newblock \emph{Advances in Neural Information Processing Systems}, 36, 2024.

\end{thebibliography}
